\documentclass[10pt,oneside]{amsart}

\usepackage{tikz}
\usetikzlibrary{automata, positioning, arrows}
\usepackage{longtable}
      \usepackage{amssymb}
\usepackage{amsfonts,amstext}
\usepackage{graphicx}
\usepackage{url}
\usepackage{comment}

\usepackage{hyperref}

\usepackage{amsmath,amstext,amssymb,amscd}
\usepackage{verbatim} 
\usepackage{mathtools}
      \usepackage{tabularx,booktabs}

\usepackage{amssymb,latexsym}
\usepackage{amsmath}
\usepackage{amsthm}
\usepackage{amssymb}
\usepackage{multicol}
\usepackage{makecell}
\newtheorem{thm}{Theorem}[section]

\newtheorem{prop}[thm]{Proposition}

\theoremstyle{remark}

\theoremstyle{definition}

\usepackage{amsthm}
 \usepackage{multicol}
\allowdisplaybreaks
\theoremstyle{definition}
\newtheorem{example}{Example}
\newtheorem{definition}{Definition}

\usepackage{multicol}

      \makeatletter
      \def\@setcopyright{}
      \def\serieslogo@{}
     
      \makeatother

\begin{document}

\author{Adam Graham-Squire}
\address{Adam Graham-Squire, Department of Mathematical Sciences, High Point University, 1 University Parkway, High Point, NC, 27268}
\email{agrahams@highpoint.edu}

\author{Matthew I. Jones}
\address{Matthew I. Jones, Yale Institute for Network Science, Yale University, 17 Hillhouse Ave, New Haven, CT, 06511}
\email{matt.jones@yale.edu}

\author{David McCune}
\address{David McCune, Department of Mathematics and Data Science, William Jewell College, 500 College Hill, Liberty, MO, 64068-1896}
\email{mccuned@william.jewell.edu}

\title[New fairness criteria for truncated ballots]{New fairness criteria for truncated ballots in multi-winner ranked-choice elections}

\begin{abstract}
In real-world elections where voters cast preference ballots, voters often provide only a partial ranking of the candidates. Despite this empirical reality, prior social choice literature frequently analyzes fairness criteria under the assumption that all voters provide a complete ranking of the candidates. We introduce new fairness criteria for multiwinner ranked-choice elections concerning truncated ballots. In particular, we define notions of the \emph{independence of losing voters blocs} and \emph{independence of winning voters blocs}, which state that the winning committee of an election should not change when we remove partial ballots which rank only losing candidates, and the winning committee should change in reasonable ways when removing ballots which rank only winning candidates. Of the voting methods we analyze, the Chamberlin-Courant rule performs the best with respect to these criteria, the expanding approvals rule performs the worst, and the method of single transferable vote falls in between.
\end{abstract}



\maketitle


\section{Introduction}

This article introduces new fairness criteria for multiwinner ranked-choice elections which stipulate how voting methods should behave when certain types of partial ballots are removed from the ballot data. Consider the following scenario: an election contains seven candidates $C_1, \dots, C_7$ and voters cast preference ballots (possibly providing only partial preferences) to determine a winning committee of size three. Suppose the winning committee under a given voting method is $\{C_1, C_2, C_3\}$ but when we remove a handful of ballots which rank only candidate $C_7$, the winning committee changes to $\{C_4, C_5, C_6\}$. It seems normatively undesirable that voters who cast ballots ranking only a losing candidate, and therefore achieve no representation, can determine the composition of the winner set by participating in the election. 
Put another way, it seems strange that candidate $C_1$'s ability to win a seat depends on the participation of voters who care only about $C_7$ (and are presumably indifferent about the other candidates). 
Similarly, suppose we remove some ballots which rank only $C_1$ and the winning committee changes to $\{C_1, C_4, C_5\}$. It seems undesirable that voters who support only $C_1$, and retain their candidate in the winner set if they abstain, could affect which other voters achieve representation. We articulate three related fairness criteria along these lines. The first is the \emph{independence of losing voter blocs (ILVB) criterion}, which stipulates that if we remove partial ballots which rank only losing candidates then the winning committee should not change. The other two criteria are analogues of the ILVB criterion but for winning candidates. The \emph{independence of winning voter blocs (IWVB) criterion} requires that if we remove partial ballots which rank only winning candidates then the only allowable change to the winning committee is that some subset of these winning candidates is replaced in the committee. Similarly, the \emph{IWVB}$^*$ \emph{criterion} stipulates that if we remove partial ballots which rank only winning candidates and all of these candidates retain their seats after this removal, then the winning committee should not change.
We provide worst-case analyses for several voting methods with respect to these criteria and give empirical results using a large dataset of real-world multiwinner elections.

Much of the previous social choice literature studies the case of single-winner elections, often under the assumption that voters provide a complete ranking of the candidates. The modern theory of fairness criteria and single-winner elections dates to the pioneering work of Kenneth Arrow \cite{A51}, which assumes complete preference information. For a summary of voting criteria in the single-winner case with complete preferences, see \cite{N99}. Some researchers have proposed criteria explicitly built around the concept of partial ballots, such as non-manipulability by sincere preference truncation \cite{B82,FB84}. Another example is Woodall's ``mono-add-plump'' criterion, which states that if candidate $A$ is the winner of an election and we add ballots which rank $A$ first and rank no other candidates, then $A$ should still win \cite{W94}. But fairness criteria explicitly concerning partial ballots tend to receive less attention in the social choice literature.

The multiwinner ranked-choice setting is less-studied than the single-winner, and correspondingly there has been less study of fairness criteria for multiwinner elections. While some single-winner criteria can be adapted to the multiwinner setting with little to no modification, many classical criteria (such as the Condorcet criterion) do not translate as easily. For an introduction to fairness criteria for multiwinner voting methods, see \cite{EFSS, FM92}. A complicating factor in the multiwinner setting is the differing goals of various multiwinner voting methods. Some aim for proportional representation and others do not, for example. We focus on multiwinner voting rules designed for proportional representation, whereas much of the prior literature focuses on axioms regarding the meaning of proportional representation \cite{AL20,BP23,D84,M95,SF17}. There are different definitions of  ``proportional''  in a multiwinner ranked-choice setting, and in Section \ref{proportionality}, we focus on Dummett's notion of proportionality for solid coalitions \cite{D84}. Our independence of  voter blocs criteria are closely related to the independence of irrelevant alternatives criterion for single-winner elections and its variants, including notions of the so-called ``spoiler effect'' \cite{A51,Bo,HC15,MW,M19}. Additionally, certain kinds of no-show paradoxes \cite{FB83} are a special case of violations of our ILVB criterion.

In this article, the primary voting method of interest is the version of single transferable vote (STV) used in Scottish local government elections because these elections are the data source for our empirical results. To give context for our results concerning Scottish STV, we also study several other proportional methods such as Meek STV~\cite{M94}, the Chamberlin-Courant voting rule \cite{Cham-Cour}, and the expanding approvals rule proposed by Aziz and Lee \cite{AL20}. In our worst-case analysis we show that, for each voting method other than Chamberlin-Courant, it is possible to construct elections in which the removal of some \emph{bullet votes} (ballots that rank a single candidate) creates a new winner set disjoint from the original winner set. For the Chamberlin-Courant rule, we prove the winning committee does not change when removing partial ballots that only rank losers, or removing ballots that rank only winners who retain their seats.

We find that in approximately 30\% of elections from our real-world Scottish dataset, at least one voting method returns a violation of our voting bloc criteria, with the percentage for a fixed voting method varying substantially. For example, about 2\% of elections return a violation of the loser bloc criterion for Meek STV, which rises to 5\% for the expanding approvals rule. The Chamberlin-Courant rule and its variants have the fewest violations while the expanding approvals rule has the most. Instead of just selecting a different candidate from the same party, these violations can have significant political ramifications. The vast majority of violations result in seats changing political parties, altering the political landscape and perhaps influencing policy outcomes. 

The paper is structured as follows. In Section \ref{prelim} we provide  all necessary preliminary information. In particular, we define our voting methods  and voter-bloc fairness criteria. We also provide a description of the real-world dataset of elections we use. In Section \ref{worst_case} we provide  worst-case analyses for violations of our criteria for all voting methods defined in Section \ref{prelim}. Section \ref{empirical_results} details our search for violations of our criteria in a large dataset of real-world multiwinner elections. In Section \ref{proportionality} we briefly discuss the relationship between proportionality and our voter bloc criteria, and Section \ref{conclusion} concludes.

\section{Preliminaries}\label{prelim}

We study only ranked-choice elections where each voter provides a (possibly partial) linear ranking of candidates using a preference ballot. Ballots are combined into a \emph{preference profile} $P$, which provides the number of ballots cast of each type. Let $m$ denote the number of candidates in an election, $V$ denote the number of voters in an election, and $k$ denote the size of the winner set, which equals the number of available legislative seats.  We refer to an ordered pair $(P,k)$ as an \emph{election}.

A \emph{multiwinner voting method} (or multiwinner voting rule) is a function which takes as input an election and outputs a winning committee of size $k$. Even though a voting method may output multiple committees due to ties, empirical examples of ties in elections are virtually nonexistent, so we avoid the issue of ties throughout our work and assume a unique winner set. We denote the winning committee under a given voting rule $W(P,k)$. The voting rule under consideration will be clear any time we use such notation, so we do not incorporate the voting method into our winner set notation.

\subsection{Multiwinner Voting Methods}

We briefly define the voting rules of interest. The first method, Scottish STV, is the method of primary interest because it is used to select the winning committee in elections from our real-world dataset.\\

\textbf{Scottish STV}.  Scottish single transferable vote (Scottish STV) is a multistage election procedure which works as follows. In a given stage, if a candidate has a number of first-place votes in excess of the  quota $\lfloor \frac{V}{k+1}\rfloor +1$, then this candidate is given a seat. Their surplus votes above quota are transferred proportionally to candidates ranked next on this candidate's ballots, where we assume that any candidate who has been previously elected or eliminated no longer appears on any ballot. If no candidate has enough first-place votes to achieve quota, then the candidate with the fewest votes is eliminated from all ballots and their first-place votes are transferred to the candidates ranked second on such ballots. One of the distinguishing features of Scottish STV, in contrast to many other forms of STV, is that once a candidate is elected they cannot receive any future vote transfers. Complete details of the method can be found at \url{https://www.legislation.gov.uk/sdsi/2007/0110714245}.\\

\textbf{Meek STV} \cite{M94}. There are two primary differences between Meek and Scottish STV. Under Meek, candidates who have previously won a seat can receive vote transfers in later rounds, and the quota can decrease as the transfer process unfolds. The exact algorithm is too complex to describe concisely; see \cite{HWW87} for a complete description.\\

\textbf{Expanding Approvals Rule (EAR)} \cite{AL20}. This method, proposed by Aziz and Lee, is designed so that voters can express weak linear orders over the candidates, but easily translates to our setting in which a voter's ordering is strict (with the possible exception that candidates left off the ballot are considered tied for the last ranking). Briefly, EAR can be thought of as a multiwinner extension of the single-winner Bucklin method. At a high level, Aziz and Lee describe their method in these terms: ``An index $j$ is initialised to 1. The voting weight of each voter is initially 1. We use a quota $q$ that is between $V/(k + 1)$ and $V/k$. While $k$ candidates have not been selected, we do the following. We perform $j$-approval voting with respect to the voters’ current voting weights. If there exists a candidate $c$ with approval support at least a quota $q$, we select such a candidate. If there exists no such candidate, we increment $j$ by one and repeat until $k$ candidates have been selected.'' In this paper, we use $q = \frac{V}{k+1}$. Ballots in the Scottish dataset are highly truncated, so for a candidate to be elected, we require that they receive a quota's worth of support only from voters who rank them on their ballot. If $j=k$ and no candidates have sufficient support because of truncated ballots, the candidate with the most support is elected. As with Meek STV, we omit a complete description of the method, and instead refer the reader to \cite{AL20}.\\

\textbf{Chamberlin-Courant Rule (CC)} \cite{Cham-Cour}. There are a number of variants of this rule. In each, voters are ``assigned'' a member of the  winning set. This assignment gives each voter a measure of individual  satisfaction; these measures are combined to create a measure of social satisfaction and the winning committee is defined to be the set of candidates that maximizes this social satisfaction.  The version of CC we consider is the original method proposed in \cite{Cham-Cour} which uses the candidates' Borda score and its sum to determine the level of individual satisfaction and social satisfaction, respectively. Formally, let $r_{ic}$ denote the rank of candidate $c$ on voter $i$'s ballot, so that this voter gives $m-r_{ic}$ points to $c$. For a fixed committee $X$ of size $k$, let $V_c(X)$ denote the set of voters for whom candidate $c$ is the most preferred candidate in the committee $X$. CC selects the committee $X$ of size $k$ which maximizes the value \[\displaystyle\sum_{c \in X}\displaystyle\sum_{i \in V_c(X)} m-r_{ic}.\]

Because CC relies on Borda scores, we must decide how many points a ballot contributes to a candidate subset if none of the candidates in that subset are ranked on the ballot. Following \cite{BFLR}, we use two implementations of CC, an ``optimistic model'' and a ``pessimistic model.'' Baumeister et al. study only single-winner Borda count in \cite{BFLR}, but their models of processing partial ballots and their terminology translate to the multiwinner CC setting. There are other ways of adapting CC to partial ballots, but these two are adequate for our purposes.\\

\textbf{Chamberlin-Courant Rule, Optimistic Model (CC OM)}. Suppose a ballot does not rank any candidates from a candidate subset currently under consideration and the ballots ranks $t$ candidates. Then this ballot gives $m-t-1$ points to the subset.\\

\textbf{Chamberlin-Courant Rule, Pessimistic Model (CC PM)}. If a ballot does not rank any candidates from a candidate subset currently under consideration, then this ballot gives no points to the subset.\\



\subsection{Our Dataset: Scottish Local Government Elections}

For the purposes of local government, Scotland is partitioned into 32 council areas, each of which is governed by a council. The councils provide a range of public services that are typically associated with local governments, such as education, waste management, and road maintenance. A council area is divided into wards, each of which elects a set number of councilors to represent the ward on the council. The number of councilors representing each ward is determined primarily by the ward’s population, but typically a ward has three or four seats. Every five years each ward holds an election where all seats available in the ward are filled using Scottish STV. Every Scottish ward has used STV for local government elections since 2007.

The data was collected for \cite{MGS24}, and is now publicly available at \url{https://github.com/mggg/scot-elex}. The dataset contains 1100 elections, of which 1070 satisfy $k>1$. Most of these elections satisfy $k \in \{3,4\}$ and $m \in \{6,7,8,9\}$; see \cite{MGS24} for more details about the dataset.

Voters are allowed to provide a complete ranking of the candidates, but voluntary truncation is very common. Approximately 14.0\% of ballots in the dataset rank only a single candidate and 58.0\% rank fewer than $k$ candidates. By contrast, only 13.2\% provide a complete ranking, where by ``complete ranking'' we mean a ballot which ranks $m$ or $m-1$ candidates. Because the ballots are so truncated in general, this dataset is a valuable resource for examining fairness criteria involving partial ballots.

\subsection{Losing and Winning Voter Bloc Criteria}

We now define our new criteria, beginning with independence of losing voter blocs.

\begin{definition}\label{def:ILVB}
Let $\mathcal{L}$ be a subset of losing candidates and $\mathcal{B}(\mathcal{L})$ be a set of ballots such that only candidates from $\mathcal{L}$ are ranked. A voting method satisfies \textbf{independence of losing voter blocs} (ILVB) if whenever we remove the ballots $\mathcal{B}(\mathcal{L})$ from the election, creating the modified profile $P'$, then $W(P,k)=W(P',k)$.
\end{definition}

The motivation behind ILVB is that if voters cast partial ballots which rank only losing candidates, these voters achieve no representation and thus removing their ballots should have no effect on the composition of the winning committee. We illustrate a violation of ILVB with a 2012 STV election from the Scottish elections database.

\begin{example}\label{first_example}

The 2012 election in the fifth ward of the East Ayrshire council area unfolded as shown in the top of Table \ref{first_ex}. The table is a ``votes by round'' table, which shows the number of votes controlled by each candidate in a given round of the STV process. We display such tables to summarize the election because the preference profile is usually much too large. Holden is by far the weakest candidate by any reasonable measure of ``weak,'' and is quickly eliminated after Knapp achieves quota. Then Scott is narrowly eliminated by 3.5 votes and the resulting winners are Knapp, Ross, and Todd.

If we remove 20 bullet votes for Holden then the election unfolds as shown in the bottom of Table \ref{first_ex}. Surprisingly, removing a handful of bullet votes for a losing candidate (and, in this case, a very weak losing candidate) causes a change in the winner set. The removal of these ballots lowers the quota by five votes, allowing for four additional votes to transfer to Scott in the second round, and subsequently Ross has the fewest votes in the third round and is eliminated. Those 20 Holden voters, with no expressed preference for Ross or Scott, were pivotal in Ross winning a seat.

\begin{table}
\begin{tabular}{c | c | c|c|c}
\multicolumn{5}{c}{Original Election, quota = 833}\\
\hline
\hline
Candidate & \multicolumn{4}{c}{Votes by Round}\\
\hline
Holden (Con)&135 & 138.7& &\\
Knapp (Lab)& \textbf{1250} & & &\\
Ross  (SNP)& 735 & 759.4 & 789.4 &\textbf{924.4}\\
Scott (Lab)& 417 &743.9&785.9&\\
Todd (SNP)& 791 & 814.4& 822.7 & \textbf{949.7}\\
\hline

\end{tabular}

\vspace{.1 in}

\begin{tabular}{c | c | c|c|c}
\multicolumn{5}{c}{Modified Election, quota = 828}\\
\hline
\hline
Candidate & \multicolumn{4}{c}{Votes by Round}\\
\hline
Holden (Con)&115 & 118.7& &\\
Knapp (Lab)& \textbf{1250} & & &\\
Ross  (SNP)& 735 & 759.6 & 789.6 &\\
Scott (Lab)& 417 &747.8&789.9&\textbf{835.4}\\
Todd (SNP)& 791 & 814.6& 823.0 & \textbf{1474.4}\\
\hline

\end{tabular}
\caption{An example of a violation of ILVB. The top table shows the vote totals in each round for the original 2012 election in the fifth ward of the East Ayrshire council area. The bottom table shows the corresponding vote totals when we remove 20 bullet votes for Holden. A bold number represents when a candidate surpasses quota.}
\label{first_ex}
\end{table}

Analysis of the full preference profile for this election arguably demonstrates additional paradoxical behavior beyond a violation of the loser blocs criterion. The election contains candidates from three parties: Conservative (Con), Labour (Lab), and the Scottish National Party (SNP). Note that it is common knowledge in Scottish politics that the Conservative party is  generally to the right of the other two. Furthermore, of the 135 voters who rank Holden first, 55 rank a Lab candidate second and 27 rank an SNP candidate second, so if we were to arrange the parties on a 1-dimensional axis then the Con candidate would be on the right, the SNP would be on the left, and Lab would be in the center. By removing bullet votes for Holden, we make the overall electorate more left-leaning, yet the winning committee becomes less left-leaning. 

\end{example}

We also note that certain kinds of no-show paradoxes are special cases of violations of the ILVB criterion. Suppose a voter casts a ballot with $A$ ranked first and $B$ ranked second, with no other candidate ranked on the ballot, and both candidates are losers in the original election. If we remove this ballot and $B$ earns a seat in the resulting winning committee, then the election exhibits a no-show paradox (because this voter is better off not voting) as well as a violation of ILVB. Such outcomes have been observed in the Scottish elections database. For example, in the 2022 election in the fifth ward of the City of Edinburgh council area candidate Malcolm Wood does not win a seat, but if we remove two bullet votes for Wood then he does win a seat \cite{MGS24}.

Our next criterion is similar to ILVB but for the removal of ballots ranking only winning candidates, and the motivation behind the criterion is the same. If, for example, $A$ is a winning candidate and we remove some bullet votes for $A$, it seems normatively undesirable that other winning candidates can lose their seats. If anything, those other winning candidates should ``better'' represent their constituencies with the removal of ballots for $A$.

\begin{definition}\label{def:IWVB}
Let $\mathcal{W}$ be a subset of winning candidates, $|\mathcal{W}|<k$, and $\mathcal{B}(\mathcal{W})$ be a set of ballots such that only candidates from $\mathcal{W}$ are ranked. A voting method satisfies \textbf{independence of winning voter blocs} (IWVB) if whenever $A$ is a winning candidate, $A \not \in \mathcal{W}$, and we remove the ballots $\mathcal{B}(\mathcal{W})$ from the election, creating the modified profile $P'$, then $A\in W(P', k)$.
\end{definition}

Notice that we can remove so many ballots that the candidates in $\mathcal W$ lose their seats without violating IWVB. However, we demonstrate below that this is perhaps too loose of a requirement to be useful, particularly for the context of proportional representation since any change in the winner set can have large ramifications for the other winners. Thus we also study a weaker version, which we call IWVB$^*$.

\begin{definition}\label{def:IWVB*}
Let $\mathcal{W}$ be a subset of winning candidates, $|\mathcal{W}|<k$, and $\mathcal{B}(\mathcal{W})$ be a set of ballots such that only candidates from $\mathcal{W}$ are ranked. A voting method satisfies \textbf{IWVB}$^*$ if whenever we remove the ballots in $\mathcal{B}(\mathcal{W})$ (creating a new profile $P'$) and $\mathcal{W}\subseteq W(P',k)$, then $W(P,k)=W(P',k)$.
\end{definition}

That is, if we remove partial ballots which rank only a subset of  winners, $\mathcal W$, and the removal of these ballots does not cause any of the candidates in $\mathcal W$ to lose their seats, then the overall winning committee should not change.

We illustrate a violation of IWVB$^*$ (and hence also a violation of IWVB) with another election from the Scottish elections database.

\begin{example}\label{second_example}
The 2022 election in the eighth ward of the North Ayrshire council area unfolded as shown in the top of Table \ref{second_ex}. There are seven candidates from six different parties: Conservative (Con), Green (Grn), Labour (Lab), Liberal Democrats (LD), Scottish Family (SFP), and the SNP. If we remove 199 bullet votes for Nairn McDonald, the Labour candidate, then the election unfolds as shown in the bottom of Table \ref{second_ex}. We might expect McDonald to lose a seat, but instead removing this unilateral support for McDonald causes Susan Johnson to replace Angela Stephen in the winner set.

\begin{table}
\begin{tabular}{c | c | c|c|c|c|c|c}
\multicolumn{8}{c}{Original Election, quota = 1007}\\
\hline
\hline
Candidate & \multicolumn{7}{c}{Votes by Round}\\
\hline
Burns (SNP)&\textbf{1470} & & & & & &\\
Collins (Grn)& 171 & 219.2 &227.6 & 239.1& 278.9& &\\
Craig (SFP)& 64 & 70.3 & 76.5 &&&&\\
Jackson (LD)& 106 & 111.7 & 130.4 & 139.9 &&&\\
Johnson (SNP)& 323 & 696.5 & 706.3 & 718.6 & 730.3 & 874.0 &\\
McDonald (Lab)& \textbf{1096}&&&&&&\\
Stephen (Con)& 795 & 797.8 & 818.5 & 839.8 & 885.3 & 913.6 & \textbf{1097.1} \\
\hline

\end{tabular}

\vspace{.1 in}

\begin{tabular}{c | c | c|c|c|c|c|c}
\multicolumn{8}{c}{Modified Election, quota = 957}\\
\hline
\hline
Candidate & \multicolumn{7}{c}{Votes by Round}\\
\hline
Burns (SNP)&\textbf{1470} & & & & & &\\
Collins (Grn)& 171 & 217.8 & 228.1 & 253.5 & 256.2 & &\\
Craig (SFP)& 64 & 69.9 &&&&&\\
Jackson (LD)& 106 & 110.5 & 117.5 &&&&\\
Johnson (SNP)& 323 & 727.1 & 738.0 & 745.7 & 747.9 & 885.6 & \textbf{975.8}\\
McDonald (Lab)& 897 & 922.1 & 929.2 & \textbf{970.6} &&&\\
Stephen (Con)& 795 & 797.8 & 818.1 & 841.8 & 846.7 & 871.1 & \\
\hline

\end{tabular}
\caption{An example of a violation of IWVB$^*$. The top table shows the original 2022 election in the eighth ward of the North Ayrshire council area. The bottom table shows the election when we remove 199 bullet votes for McDonald.}
\label{second_ex}
\end{table}

We make two observations about this example. First, with Scottish STV, a violation of a voter bloc criterion can manifest in different ways. In Example \ref{first_example}, the actual election and the modified election unfold in the same sequential order until the final round. That is, removing bullet votes for Holden did not change the order in which the election unfolded in intermediate rounds: in either election Knapp is elected in the first round, Holden is eliminated in the second round, etc.  By contrast, in Example \ref{second_example} the removal of votes for McDonald causes the election to unfold in a different order, so that McDonald is not elected until the fourth round in the modified election. This change in the order of how the election unfolds ultimately leads to a violation of IWVB$^*$. Second, the party dynamics in this election demonstrate why IWVB might be normatively desirable. It seems strange that the SNP can double its representation in the winning committee if voters who care only about one Labour candidate decide not to vote. The voters who cast bullet votes for McDonald are presumably indifferent between all other candidates; why should their votes determine if a Conservative or SNP candidate earns the third seat? If anything, the presence of these left-leaning voters should support the election of a left-leaning candidate from the SNP, but instead these Labour voters are necessary for a Conservative to win the final seat.

\end{example}

When introducing new fairness criteria, our first question should be: do any reasonable voting methods satisfy these criteria? In our case the answer is Yes, since any positional scoring rule such as $k$-Borda or $k$-plurality satisfies both ILVB and IWVB (the reason is that a candidate's score does not decrease when removing ballots on which that candidate is not ranked), but these methods do not aim to achieve proportional representation. As we discuss below, depending on what is meant by ``proportional,'' proportional methods which satisfy ILVB and IWVB$^*$ do exist, but we are unable to find a ranked-choice proportional method which satisfies IWVB. Because positional scoring rules satisfy ILVB and IWVB, these are distinguished from more classical criteria such as Arrow's independence of irrelevant alternatives, which no reasonable method satisfies.

\section{Worst Case Analyses}\label{worst_case}

In this section, we show that for each method except CC we can change the entire winning committee by removing partial ballots with only losing candidates ranked on them. The same is true for every method including CC when we remove partial ballots with only winning candidates ranked on them. Thus, the worst-case for both ILVB and IWVB violations is as bad as possible, with the exception of ILVB with CC. We also show that CC satisfies the IWVB$^*$ criterion, so the worst-case analysis does not apply to CC in this case. Throughout this section, we use a quota of $q=\frac{V}{k+1}$ for convenience. While some of the methods we described in Section \ref{prelim} use slightly different quotas such as $\lfloor \frac{V}{k+1} \rfloor + 1$, the differences are vanishingly small for large $V$ and have no impact on the theoretical results presented here.

The preference profiles used to demonstrate worst-case outcomes are extreme, and unlikely to be observed in practice. In Section \ref{empirical_results}, we explore the outcomes found in real-world elections.

\begin{prop}
(ILVB - STV and EAR) Let $k>0$. For Scottish STV, Meek STV, and EAR, there exists a profile $P$ (which depends on the voting method) and a set of ballots from $P$ ranking only losing candidates such that when we remove these ballots, creating the modified profile $P'$, $W(P,k)\cap W(P',k)=\emptyset$. 
\end{prop}

\begin{proof}

Fix $k>0$. For Scottish and Meek STV, consider the profile outlined in Table \ref{tab:scot_stv_losers}. (This profile is inspired by Example 3.5.2.1 in \cite{F12}.) The profile contains $3k$ candidates which we label $A_1,\dots, A_k, B_1, \dots, B_k, C_1, \dots, C_k$. For each $1\leq i \le k$, let there be 7 ballots of the form $A_i>B_i>C_i$, 9 ballots of the form $A_i>C_i>B_i$, 12 ballots of the form $B_i>C_i>A_i$, 13 ballots of the form $C_i>A_i>B_i$, and 2 bullet votes for $B_i$. No candidate achieves quota initially, and under either form of STV the $C_i$ candidates are all eliminated, resulting in the winning committee $\{A_1,\dots, A_k\}$. If we remove the 2 bullet votes for each $B_i$, then the $B_i$ candidates are eliminated, resulting in the winning committee $\{C_1, \dots, C_k\}$.

\begin{table}[h]
    \centering
    \begin{tabular}{c|c|c|c|c|c|c}
        $7$ & $9$ & $12$ & $13$ & $2$ & $7$ & $\dots$ \\
        \hline \hline
        $A_1$ & $A_1$ & $B_1$ & $C_1$ & $B_1$ & $A_2$ & $\dots$ \\
        \hline
        $B_1$ & $C_1$ & $C_1$ & $A_1$ & & $B_2$ & $\dots$ \\
        \hline
        $C_1$ & $B_1$ & $A_1$ & $B_1$ & & $C_2$ & $\dots$
        
    \end{tabular}
    \caption{A profile illustrating a worst-case outcome for Scottish or Meek STV with respect to ILVB. If there are $k$ seats, create $k$ copies of the first five columns, with each copy containing a disjoint set of candidates.}
    \label{tab:scot_stv_losers}
\end{table}

For EAR, consider the profile containing $2k+1$ candidates outlined in Table \ref{tab:exp_app_losers}.  First, when the bullet votes for $C$ are present, the quota is $q = \frac{(10k+8)k+3k}{k+1} = \frac{10k^2 + 11k}{k+1}$. In the first round, the $A_i$ candidates have 8 votes and the $B_i$ candidates have $10k$ votes. Note $q = \frac{10k^2 + 11k}{k+1} > \frac{10k^2 + 10k}{k+1} = 10k$ and therefore no candidate initially achieves quota. In the second round, all $A_i$ candidates have $10k+8$ votes. Because $q = \frac{10k^2 + 11k}{k+1} < \frac{10k^2 + 18k + 8}{k+1} = 10k + 8$, all $A_i$ candidates achieve quota and the winning committee is $\{A_1, \dots, A_k\}$. 

If we remove the $3k$ bullet votes for $C$, then the quota is $q = \frac{(10k+8)k}{k+1} = \frac{10k^2 + 8 k}{k+1} < \frac{10k^2 + 10k}{k+1} = 10k$. Thus, the $B_i$ candidates achieve quota in the first round and the winning committee is $\{B_1, \dots, B_k\}$.

\begin{table}[h]
    \centering
    \begin{tabular}{c|c|c|c|c|c}
        $8$ & $10k$ & $\dots$ & $8$ & $10k$ & $3k$ \\
        \hline \hline
        $A_1$ & $B_1$ & $\dots$ & $A_k$ & $B_k$ & $C$ \\
        \hline
        & $A_1$ &  &  & $A_k$ & 
        
    \end{tabular}
    \caption{An example of an election where the winner set completely changes according to the Expanding Approvals Rule when throwing out the $C$ loser ballots.}
    \label{tab:exp_app_losers}
\end{table}

\end{proof}

While the worst-case ILVB violations for STV and EAR are as extreme as possible, by contrast the Chamberlin Courant rule satisfies the criterion.

\begin{prop}
CC OM and CC PM satisfy ILVB.
\end{prop}

\begin{proof}
Let $W(P,k)$ be the winner set under either version of CC. Let $\mathcal{B}(\mathcal{L})$ be a set of ballots which rank only losing candidates. If we remove the ballots from $\mathcal{B}(\mathcal{L})$, then the CC score for the set $W(P,k)$ may decrease (in the OM model, e.g.) but any set of candidates not containing one of the losing candidates from $\mathcal{B}(\mathcal{L})$ will lose the same amount of points as $W(P,k)$. By definition of the CC rule, any set of candidates containing one of the losing candidates ranked on the ballots in $\mathcal{B}(\mathcal{L})$ will lose even more points than $W(P,k)$. Therefore, in the modified profile with the ballots in $\mathcal{B}(\mathcal{L})$ removed, the set $W(P,k)$ still has the maximal CC score and remains the winner set.
\end{proof}

We note that any CC rule which uses a reasonable model to process partial ballots also satisfies ILVB.

We now prove an analogous proposition for IWVB. In this case, CC also produces worst-case outcomes.

\begin{prop}
(IWVB) Let $k>0$. For Scottish STV, Meek STV, EAR, and both versions of CC, there exists a profile $P$ (which depends on the voting method) and a set of ballots from $P$ which rank only a single winning candidate such that when we remove these ballots, creating the modified profile $P'$, $W(P,k)\cap W(P',k)=\emptyset$. 
\end{prop}

\begin{proof}

Fix $k>0$. For Scottish and Meek STV, consider the profile with $2k$ candidates outlined in Table \ref{tab:scot_stv_winners}. The quota is $20k - 12$, which also equals the number of first-place votes for $A_1$. No other candidate achieves quota initially, and thus $A_1$ earns the first seat but has no surplus to transfer. $B_1$ has the fewest first place votes and is eliminated, transferring two votes to $A_i$ for $2\le i \le k$. As a result, each of the remaining $B_i$ candidates have fewer votes than any of the $A_j$ candidates, all $B_i$ candidates are eliminated, and the winning committee is $\{A_1, \dots, A_k\}$.

However, if we remove the $14k-12$ bullet votes for $A_1$, then $A_1$ has only $6k$ votes and is eliminated in the first round. Each $B_i$ candidate receives a positive vote transfer from $A_1$ while the remaining $A_i$ candidates do not. Thus each of the $A_i$ candidates are eliminated and the winning committee is $\{B_1, \dots, B_k\}$.

\begin{table}[h]
    \centering
    \begin{tabular}{c|c|c|c|c|c|c|c|c|c|c|c|c|c}
        $14k - 12$ & $4k + 2$ & $2$ & $\dots$ & $2$ & $6k+2$ & $2$ & $\dots$ & $2$ & $10k$ & $10k$ & $\dots$ & $10k$ & $10k$ \\
        \hline \hline
        $A_1$ & $A_1$ & $A_1$ & $\dots$ & $A_1$ & $B_1$ & $B_1$ & $\dots$ & $B_1$ & $A_2$ & $B_2$ & $\dots$ & $A_k$ & $B_k$ \\
        \hline
         & $B_1$ & $B_2$ & $\dots$ & $B_k$ & & $A_2$ & $\dots$ & $A_k$ & & & & & 
        
    \end{tabular}
    \caption{An example of a profile where the winner set completely changes when throwing out the $A_1$ bullet votes using Scottish STV or Meek STV.}
    \label{tab:scot_stv_winners}
\end{table}

For EAR, consider the profile in Table \ref{tab:exp_app_winners}. The quota is $q=20k+10$, so candidate $A$ is the only candidate to earn quota initially and is given the first seat. Furthermore, because quota is so large, no other candidates make quota, and all candidates get support from any ballot that they appear on. When considering the second place votes, the $B_i$ candidates have $10 + 20k$ votes and make quota. The $B_i$ candidates win the remaining $k-1$ seats and the winning committee is $\{A,B_1, \dots, B_{k-1}\}$.

If we remove the bullet votes for $A$, the quota is $q = \frac{10(2k+1)(k-1)}{k+1}$, which is less than $20k$. Thus, the $C_i$ candidates win the $k$ seats and the winning committee changes to $\{C_1, \dots, C_{k}\}$.

\begin{table}[h]
    \centering
    \begin{tabular}{c|c|c|c|c|c|c|c}
        $20k+20$ & $10$ & $\dots$ & $10$ & $20k$ & $\dots$ & $20k$ & $20k$ \\
        \hline \hline
        $A$ & $B_1$ & $\dots$ & $B_{k-1}$ & $C_1$ & $\dots$ & $C_{k-1}$ & $C_k$ \\
        \hline
        & & & & $B_1$ & $\dots$ & $B_{k-1}$ &
        
    \end{tabular}
    \caption{An example of a profile where the winner set completely changes under EAR when throwing out the $A$ bullet ballots.}
    \label{tab:exp_app_winners}
\end{table}

For CC, consider the profile in Table \ref{tab:cpo_stv_winners}.  For $1<i<k$, there are two ballots of the form $A_i>B_i$, two of the form $B_i>A_i$, two of the form $A_i>B_{i+1}$, and two of the form $B_{i+1}>A_i$. For $i \in \{1,k\}$ the ballots are as shown in the table. Let $\mathcal{A}=\{A_1, \dots, A_k\}$ and $\mathcal{B}=\{B_1, \dots, B_k\}$. For CC OM, the score for $\mathcal{A}$ is \[(2k-1)(5+4(k-1)) + (2k-2)(4k) = 16k^2-10k-1\]  and the score for $\mathcal{B}$ is \[(2k-1)(4k) + (2k-2)(3+2+4(k-1)) = 16k^2 - 10k - 2.\] Thus, $\mathcal{A}$ will have more CC OM points than $\mathcal{B}$ in this case. In the pessimistic model, the score for $\mathcal{B}$ goes down by $3(2k-2)$ to $16k^2 - 16k + 4$ so $\mathcal{A}$ still beats $\mathcal{B}$ in that situation.  Since $\mathcal{A}$ has the property that every voter has a candidate in the subset ranked in their top two rankings, swapping out any $A_i$ with a $B_j$ will result in a lower CC score.  Specifically, exchanging an $A_i$ with a $B_j$, will result in some columns of votes that have $A_i$ now receiving no points (or only $2k-3$ points in OM), leading to a reduced CC score.  Exchanging any $B_i$ for $A_j$ in $\mathcal{B}$ similarly results in a reduced CC score.  It follows that $\mathcal{A}$ wins both CC OM and CC PM for the profile in Table \ref{tab:cpo_stv_winners}.

Now consider the profile in Table \ref{tab:cpo_stv_winners} with the first column (3 bullet votes for $A_1$) removed, and call this profile $P'$.  For $P'$, the score for $\mathcal{A}$ is $(2k-1)(2+4(k-1)) + (2k-2)(4k) = 16k^2-16k+2$ and the score for $\mathcal{B}$ is $16k^2 - 16k + 4$ (equal to the PM score with bullet votes), so $\mathcal{B}$ has more points than $\mathcal{A}$.  As with the original profile, swapping out some candidates from $\mathcal{A}$ or $\mathcal{B}$ will similarly result in lower scores for the modified subsets, thus $\mathcal{B}$ wins both CC OM and CC PM for profile $P'$. 

\begin{table}[h]
    \centering
    \begin{tabular}{c|c|c|c|c|c|c|c|c|c|c|c|c|c|c}
        $3$ & $1$ & $2$ & $1$ & $2$ & & $2$ & $2$ & $2$ & $2$ & & $2$ & $2$ & $2$ & $2$ \\
        \hline \hline
        $A_1$ & $A_1$ & $B_1$ & $A_1$ & $B_2$ & & $A_i$ & $B_i$ & $A_i$ & $B_{i+1}$ & & $A_k$ & $B_k$ & $A_k$ & $B_1$ \\
        \hline
        & $B_1$ & $A_1$ & $B_2$ & $A_1$ & & $B_i$ & $A_i$ & $B_{i+1}$ & $A_i$ & & $B_k$ & $A_k$ & $B_1$ & $A_k$
        
    \end{tabular}
    \caption{An example of a profile where the winner set completely changes according to both CC rules (OM and PM) when throwing out winner ballots. The middle block of ballots are copied for $i = 2, 3, \dots k-1$.}
    \label{tab:cpo_stv_winners}
\end{table}
\end{proof}

To conclude this section, we provide the worst-case analysis for violations of IWVB$^*$ for the STV and EAR methods, and show that both forms of CC satisfy the criterion.

\begin{prop}
(IWVB$^*$ - STV and EAR) Let $k>0$. For Scottish STV, Meek STV, and EAR, there exists a profile $P$ (which depends on the voting method) and a set of ballots from $P$ which rank only a single winning candidate $A$ such that when we remove these ballots, creating the modified profile $P'$, $W(P,k) \cap W(P', k) = \{A\}$.
\end{prop}

\begin{proof}
We first give the proof in full generality, then give an illustrative example at the end. Fix $k>0$. For Scottish and Meek STV, consider the profile $P$ in Table \ref{tab:scot_stv_winners_alt}.

\begin{table}[h]
    \centering
    \begin{tabular}{c|c|c|c|c|c|c}
        $a$ & $b$ & $\dots$ & $b$ & $c$ & $\dots$ & $c$ \\
        \hline \hline
        $A$ & $A$ & $\dots$ & $A$ & $C_1$ & $\dots$ & $C_{k-1}$ \\
        \hline
         & $B_1$ & $\dots$ & $B_{k-1}$ & &  & 
        
    \end{tabular}
    \caption{A generic set of ballots that demonstrate the worst-case scenario for STV. $A$ is always a winner. For sufficiently large $a$, when we throw out the $A$ bullet votes, the rest of the winner set will change from $B$s to $C$s if $\frac{kb}{k+1}>c>\frac{kb-c}{k+1}$.}
    \label{tab:scot_stv_winners_alt}
\end{table}

For large enough $b$ (relative to $c$), $A$ will clearly win the first seat. After $A$ is elected, either the $B_i$ candidates have more votes and the $C_i$ candidates are removed, or vice versa, depending on the relative values of $a$, $b$ and $c$. We must determine how many votes each $B_i$ candidate receives. The quota is $\frac{a + (k-1)(b+c)}{k+1}$, so each $B_i$ candidate receives 

\begin{equation*}
\frac{a+(k-1)b - \frac{a + (k-1)(b+c)}{k+1}}{a+(k-1)b}b.
\end{equation*}

When $a=0$, the $B_i$ candidates will receive the fewest votes possible with $\frac{kb-c}{k+1}$ votes. On the other hand, if we take the limit as $a \to \infty$, the share of $B$ votes approaches $\frac{kb}{k+1}$.

Therefore, if $\frac{kb}{k+1} > c > \frac{kb-c}{k+1}$, there is a value of $a$ where $A$ transfers enough ballots to the $B_i$ candidates that they also win seats, but if the $a$ ballots are removed, $A$ consumes almost all the votes and the $C_i$ candidates win. One example that works for Scottish and Meek STV is in Table \ref{tab:scot_stv_winners_alt_example} with $k = 3$, $a = 1000$, $b = 20$, and $c = 13$. 

\begin{table}[h]
    \centering
    \begin{tabular}{c|c|c|c|c}
        1000 & 20  & 20 & 13 & 13 \\
        \hline \hline
        $A$ & $A$ & $A$ & $C_1$ & $C_2$ \\
        \hline
         & $B_1$ & $B_{2}$ & & 
        
    \end{tabular}
    \caption{An election demonstrating a worst-case scenario for both versions of STV. When  the 1000 bullet votes for $A$ are removed, the winner set changes from $\{A, B_1, B_2\}$ to $\{A, C_1, C_2\}$.}
    \label{tab:scot_stv_winners_alt_example}
\end{table}

\begin{table}[h]
    \centering
    \begin{tabular}{c|c|c|c|c|c|c|c|c|c}
        $a$ & $10k$ & $\dots$ & $10k$ & $10$ & $\dots$ & $10$ & $10$ & $\dots$ & $10$\\
        \hline \hline
        $A$ & $A$ & $\dots$ & $A$ & $B_1$ & $\dots$ & $B_{k-1}$ & $C_1$ & $\dots$ & $C_{k-1}$ \\
        \hline
        & $C_1$ & $\dots$ & $C_{k-1}$ & & & & $B_1$ & $\dots$ & $B_{k-1}$ \\
        \hline
        & $B_1$ & $\dots$ & $B_{k-1}$ & & & & & &
    \end{tabular}
    \caption{An example of a set of ballots. $A$ is always a winner. For sufficiently large $a$, when we throw out the $A$ bullet votes, the rest of the winner set will change from $C$s to $B$s}
    \label{tab:exp_app_winners_alt}
\end{table}

For EAR, consider Table \ref{tab:exp_app_winners_alt}. When there are no bullet votes for $A$ ($a = 0$), quota is $q = \frac{(10k+20)(k-1)}{k+1} = 10\frac{(k+2)(k-1)}{k+1}$. $A$ makes quota, so $A$ wins a seat and the ballots with $A$ ranked first have their weights adjusted from 1 to $\epsilon = \frac{10k(k-1) - q}{10k(k-1)} = \frac{k^2 - 2}{k(k+1)}$. No other candidates make quota so the second place votes are considered. The $B_i$ candidates have $20$ votes, which is short of quota, but the $C_i$ candidates have $10 + 10 \epsilon k$. $10 + 10 \epsilon k = 10 + 10 \frac{k^2 - 2}{k+1} = 10\frac{k + 1 + k^2 - 2}{k+1} = 10\frac{k^2 + k - 1}{k+1} > 10 \frac{(k+2)(k-1)}{k+1} = q$, so the $C_i$ candidates all make quota and are elected, filling all $k$ seats.

Now consider what happens as $a$ approaches infinity. As $a\to \infty$, $q\to \infty$, but $A$ still makes quota in the first round.  Once $A$ is removed and votes transferred, though, the ballots with $A$ ranked first have weights that approach $\frac{k}{k+1}$, and no other candidates ever make quota. After considering the second and then third place votes, the $B_i$ candidates have $20+10k\frac{k}{k+1}$ votes, which is greater than the $10 + 10k\frac{k}{k+1}$ votes the $C_i$ candidates have. Therefore, the $B$ candidates win seats. \end{proof}

\begin{prop}
CC OM and CC PM satisfy IWVB$^*$.
\end{prop}



\begin{proof}
    Let $W$ be the winner set under a given version of CC with the ballot profile $P$. Let $\mathcal{B}$ be a set of ballots which rank only candidates in $\mathcal W \subset W$. We claim that if we remove the ballots in $\mathcal B$ from $P$, creating a modified profile $P'$ which has the (potentially different) winner set $W'$ with $\mathcal W \subset W'$, then $W = W'$.

    Suppose not, so $W \neq W'$. Because $W$ is the winning committee under $P$, the CC score for $W$ is greater than the score for $W'$ under $P$. When we remove the ballots in $\mathcal B$, this subtracts the same number of points from both $W$ and $W'$ because $\mathcal W$ is a subset of both $W$ and $W'$. Therefore, under $P'$, the score for $W$ must still be greater than the score for $W'$. Therefore, $W'$ cannot be the winning committee under $P'$, a contradiction, so we conclude that $W = W'$.
\end{proof}

\section{Empirical Results}\label{empirical_results}

In this section we provide empirical results using the Scottish elections dataset. Before presenting the results, we describe the methodology used to find violations of the criteria.

\subsection{Methodology for finding violations}

As the number of ballots cast in an election grows, the number of different sets of ballots that could be removed to potentially cause a violation grows exponentially, and thus any kind of brute force algorithm is not feasible. Therefore, we use a heuristic search that removes only a small subset of all possible ballot combinations to test for violations of our voter bloc criteria. The procedure for searching for losing voter bloc violations is slightly different than winning voter bloc violations, but they both begin by finding the winner set $W$ and the loser set $L$ of the full election. We give high level descriptions of these algorithms below. The code we use, along with all violations found under each method (including a record of the ballots removed) can be found at https://github.com/MattJonesMath/Irrelevant\_Voters\_Project. We note that some of the code we use for the EAR method was provided by Jannik Peters.

Because we use heuristic algorithms, we cannot guarantee we found all criteria violations for any of the voting methods. With that said, we implemented code at different levels of granularity and stopped when we found only marginal changes in the number of violations.  Due to this, we believe we have found almost all of the ILVB, ILWB, and ILWB$^*$ violations present in the Scottish data. Ideally we could find necessary and sufficient conditions that a preference profile must satisfy to produce a violation for a given method, and our code could simply check each preference profile against such conditions. Given the challenge of finding such conditions for various fairness criteria in the three-candidate single-winner case \cite{F19, GS24, Mi17} and the absence of any known conditions in the four-candidate case, we doubt such conditions are forthcoming in our setting.

\subsubsection{ILVB Violation Search}

Select a candidate $B \in L$. We try to remove loser ballots so that $B$ wins a seat. It is very unlikely (although not impossible) that removing a ballot that ranks $B$ will help $B$ win a seat, so identify $\mathcal B$, the set of all ballots that rank only losing candidates that are not $B$.

Removing loser ballots has no impact on the votes that winning candidates receive. Instead, removing these ballots can influence the winner set by changing the quota. We want to test many different quotas, because violations are not monotone with respect to the removal of ballots. That is, it is possible that removing some fraction of the ballots in $\mathcal B$ causes a violation but if we remove all the ballots in $\mathcal B$, we do not observe a violation. Therefore, we try removing $\sigma_\ell$ different fractions ($\frac{1}{\sigma_\ell}, \frac{2}{\sigma_\ell}, \dots, 1)$ of $\mathcal B$, each time rerunning the election to see if $B$ has won a seat. For neatness, if this requires removing a fractional ballot, we round down the number of ballots removed to the next largest integer.

Repeat this process for all $B \in L$.

\subsubsection{IWVB and IWVB$^*$ Violation Search}

Select a pair of candidates $A \in W$ and $B \in L$. We try to remove winner ballots so $B$ takes $A$'s seat. Because we are removing winner ballots now, these ballots can influence the election by changing the quota, but also by preventing surplus votes from being transferred from one candidate to another. Therefore, in this search, we consider how many votes each winning candidate who is not $A$ may send to $A$ and $B$ if they are elected. For each candidate $C_* \in W \setminus \{A\}$, count the number of ballots that rank $C_*$ above both $A$ and $B$ but rank $A$ above $B$. These are ballots that could potentially be transferred to $A$ if $C_*$ wins a seat. Then also count the number of ballots that rank $C_*$ above both $A$ and $B$ but rank $B$ above $A$, which are ballots that could potentially transfer to $B$. Rank the candidates in $W \setminus \{A\}$ by the difference of these two counts so we have $C_1$, $C_2$, $\dots$, $C_{k-1}$. $C_1$ will have many votes that could transfer to $A$ and $C_{k-1}$ will have many votes that could transfer to $B$.

We begin by taking $\mathcal B$ to be only bullet votes for $C_1$. Like above, we want to try many different quota values, so remove $\sigma_w$ different fractions of $\mathcal B$ and try running the election.

Continue by letting $\mathcal B$ be the ballots that only rank $C_1$ and $C_2$, then $C_1$ through $C_3$, and so on until $\mathcal B$ is all ballots that rank all the winning candidates except $A$, each time removing $\sigma_w$ fractions of $\mathcal B$. 

Repeat this process for all pairs of $A \in W$ and $B \in L$.

When searching for violations for the alternate definition IWVB$^*$ (Def \ref{def:IWVB*}), we also confirm that all candidates ranked by a ballot in $\mathcal B$ still win a seat.

\subsubsection{Party Dynamics}

When searching for violations that result in a change in party control, we modify the above algorithms as follows. Select a pair of candidates $A \in W$ and $B \in L$ (even for the ILVB search). When considering which ballots to remove, restrict to only ballots that do not rank any candidates from the parties of $A$ and $B$. If a potential violation has been found where $A$ loses their seat and $B$ gains a seat, confirm that $A$'s party has one fewer seat than before, and $B$'s party has one additional seat. This is to avoid the case where many winning ballots are removed and many seats change winners, which can make it appear that a seat has changed parties even if it has not.

\subsection{Violations in Scottish Elections}

We searched for violations in the Scottish election data with the five  methods from Section \ref{prelim}. We found 576 different violations (not counting the IWVB$^*$ violations, which are special cases of IWVB violations) in 325 elections. Our algorithms could find no violations under any of our voting methods in 745 elections ($\approx 70\%$). CC OM was least susceptible to violations, followed closely by CC PM. As expected, we found no ILVB or IWVB$^*$ violations in the Scottish data using a CC rule. Meek STV is the sequential process with the fewest violations, although Scottish STV has slightly fewer IWVB$^*$ violations. The expanding approvals rule had the most of every type of violation. The full results for $\sigma_\ell = 10$ and $\sigma_w = 3$ are shown in Table \ref{tab:scot_results}. Larger values of $\sigma_\ell$ and $\sigma_w$ yield only marginal improvements for large increases in computational cost. 

By examining which elections have violations for specific methods, we confirm that there are no violations with one method that guarantee a violation with another method. It seems likely that there are no relationships of this kind, although we cannot say that for certain since we did not conduct an exhaustive search for all possible violations. The closest we come in this data is that IWVB violations under CC OM only occur in elections in which there is also an anomaly with either CC PM or EAR.

\begin{table}[]
    \centering
    \begin{tabular}{l||c|c|c|c|c|}
        & Scottish & Meek & EAR & CC OM & CC PM \\
        \hline
        ILVB Violations & 40 & 19 & 54 & 0 & 0 \\
        \hline
        IWVB Violations & 109 & 104 & 199 & 23 & 28 \\
        \hline
        IWVB$^*$ Violations & 104 & 103 & 181 & 0 & 0
    \end{tabular}
    \caption{The number of violations found in the Scottish election data using parameters $\sigma_\ell = 10$ and $\sigma_w = 3$ in our violation search. The first row shows how many elections we found that violate ILVB (Def \ref{def:ILVB}), the second row counts the elections that violate IWVB (Def \ref{def:IWVB}), and the third row counts elections that violate the alternative IWVB$^*$ (Def \ref{def:IWVB*}).}
    \label{tab:scot_results}
\end{table}

Finally, for about 10\% of found anomalies, the violation requires throwing out only a fraction of all possible ballots. When $\sigma_\ell = \sigma_w = 1$ (where we throw out all of $\mathcal B$ and do not try different quota values), we find about 10\% fewer violations, as can be seen by comparing Tables \ref{tab:scot_results} and \ref{tab:scot_results_simple_search}. EAR seems to contain the bulk of the non-monotone violations.

\begin{table}[]
    \centering
    \begin{tabular}{l||c|c|c|c|c|}
        & Scottish & Meek & EAR & CC OM & CC PM  \\
        \hline
        ILVB Violations & 38 & 19 & 48 & 0 & 0  \\
        \hline
        IWVB Violations & 94 & 93 & 164 & 21 & 25  \\
        \hline
        IWVB$^*$ Violations & 89 & 92 & 142 & 0 & 0 
    \end{tabular}
    \caption{The number of violations found in the Scottish election data when $\sigma_\ell = \sigma_w = 1$. Compare to Table \ref{tab:scot_results} to see how many violations are non-monotone and require keeping some ballots and removing others.}
    \label{tab:scot_results_simple_search}
\end{table}

\subsection{Party Dynamics}

In slightly more than half of elections with a violation, a seat can change parties without discarding any ballots that support either party. Those results are shown in Table \ref{tab:scot_results_parties}.

\begin{table}[]
    \centering
    \begin{tabular}{l||c|c|c|c|c|}
        & Scottish & Meek & EAR & CC OM & CC PM  \\
        \hline
        ILVB Party Swaps & 32 & 13 & 42 & 0 & 0 \\
        \hline
        IWVB Party Swaps & 64 & 60 & 126 & 11 & 15 \\
        \hline
        IWVB$^*$ Party Swaps & 62 & 60 & 113  & 0 & 0
    \end{tabular}
    \caption{The number of violations found in the Scottish election data where at least one seat changes party when $\sigma_\ell = 10$ and $\sigma_w = 3$. Compare to Table \ref{tab:scot_results} to see that slightly more than half of violations can result in changes to the party composition of the final committee.}
    \label{tab:scot_results_parties}
\end{table}

For the purposes of these results, we treat all independent candidates as members of one party. Therefore, these results represent a conservative estimate of how often these violations can affect the political composition of a council. Two independent candidates could have completely opposing views, but we would not notice that shift with the data we have.

\section{Proportionality and Voter Bloc Criteria}\label{proportionality}

Throughout this article we focus on methods that aim to provide proportional representation. As we have seen, with the exception of CC the proportional methods we study fail all of our proposed criteria. Perhaps these results point to a conflict between achieving proportionality and satisfying the irrelevant voter bloc criteria. In this section, we explore such potential tension.

First, we briefly review what is commonly meant by ``proportionality.'' In the  social choice literature, proportionality is commonly understood in terms of Dummett's proportionality for solid coalitions (PSC) criterion \cite{D84} and its variants \cite{W94}. Put simply, these criteria first establish a quota $q \in (\frac{V}{k+1}, \frac{V}{k}]$. The $q$-PSC criterion states that if a solid coalition of voters achieves a size of at least $jq$ for some $j \in \mathbb{N}$, then at least $j$ candidates supported by the coalition should earn seats, assuming the coalition supports at least $j$ candidates. If the coalition supports fewer than $j$ candidates, then $q$-PSC requires that all candidates supported by the coalition earn seats. (A set of voters form a \emph{solid coalition} if the candidates can be partitioned into two sets $\mathcal{C}_1$ and $\mathcal{C}_2$ so that every voter in the coalition prefers every candidate in $\mathcal{C}_1$ to every candidate in $\mathcal{C}_2$.)  PSC has been referred to as the most important requirement for proportional representation \cite{T95}, and is often used as the justification for using some form of STV. (All forms of STV used in this article, as well as EAR, satisfy $q$-PSC for the Hare quota $q=\frac{V}{k}$. CC does not satisfy $q$-PSC for any $q$, but in the Scottish election data, there are no violations for $q = \frac{V}{k}$.) There is a sizable literature which studies PSC and proportionality more broadly, and we do not attempt to survey it here. We refer the reader to \cite{AL20, D84, BP23, EFSS, W94} for deeper discussions of PSC and related criteria.

Practically speaking, in a given election the $q$-PSC criterion establishes a set of acceptable winning committees from which a $q$-PSC-compatible voting rule must choose. If a candidate earns more than $q$ first-place votes then any PSC-acceptable winning committee includes that candidate, for example. This requirement illustrates why there is tension between satisfying $q$-PSC and a criterion involving irrelevant voter blocs. If we remove enough ballots from an election, then $q$, which is a function of the number of voters, decreases. Consequently, some solid coalitions which were too small to make claims on seats with the original quota may now be large enough to earn seats with the smaller quota. That is, if we reduce the quota enough, then the resulting set of PSC-acceptable winning committees is a proper subset of the PSC-acceptable committees in the original election. As a result, satisfying PSC may require a new winning committee if the original winning committee is no longer in the set of PSC-acceptable committees after quota is decreased.

To demonstrate the challenge of trying to satisfy both PSC and irrelevant voter bloc criteria, and to illustrate the observations of the previous paragraph, we introduce a new family of PSC-compatible voting methods based on positional scoring. As mentioned in Section \ref{prelim}, all positional scoring rules satisfy our independence of irrelevant voter bloc criteria, and therefore adapting such rules to the PSC setting is a natural test case for trying to reconcile PSC and irrelevant voter bloc criteria. Furthermore, these new voting methods do not occur sequentially in rounds like most classical methods which satisfy PSC, and thus we can analyze our criteria independent of the sequential nature of STV-type voting rules.

A \emph{scoring rule} is a method which uses a scoring vector $(s_1, \dots, s_m)$ with $s_i\ge s_{i+1}$, $s_1>0$, and $s_i\ge 0$ to determine the number of points that a ballot contributes to a candidate's score, based on where that candidate is ranked on the ballot. If the ballot ranks a candidate in the $i$th ranking, then the candidate receives $s_i$ points from the ballot. The $k$ candidates with the largest total scores, summed across all ballots, are the winners.  Consider a modification of a scoring rule which is designed to satisfy $q$-PSC.

\begin{definition}
Fix $q \in (\frac{V}{k+1}, \frac{V}{k}]$ and a scoring vector $(s_1, \dots, s_m)$. The $q$-\emph{PSC scoring rule} works as follows. For a given preference profile, determine all candidate subsets of size $k$ which are compatible with $q$-PSC (this is computationally straightforward; see \cite{AL22}, for example). Assign each of these subsets a score by adding the scores, as determined by $(s_1, \dots, s_m)$, for each candidate in the subset. The subset with the highest score is the winning committee. For simplicity, we assume partial ballots are processed under the pessimistic model, so that candidates left off a ballot receive no points from it.
\end{definition}

That is, a $q$-PSC scoring rule is simply a positional scoring rule which can consider only winning committees compatible with $q$-PSC. We note that the following analysis can be easily adapted to other models of processing partial ballots.

Suppose we set $q$ just larger than $V/(k+1)$, as in STV. We prove that any $q$-PSC scoring rule violates ILVB (it is straightforward to show violation of IWVB as well).

\begin{prop}
Let $q=\lfloor V/(k+1)\rfloor +1$. Then any $q$-PSC scoring rule violates ILVB.
\end{prop}

\begin{proof}

\begin{table}[h]
    \centering
    \begin{tabular}{cccc}

        \begin{tabular}{c|c|c|c}
        333&1&333&332\\
        \hline
        \hline
        $A$&$B$&$C$&$D$\\
         & & $D$&$C$\\
        \end{tabular}

        &&
        \begin{tabular}{c|c|c}
        1&666&332\\
        \hline
        \hline
        $A$&$C$&$B$\\
        & $D$ &\\
        \end{tabular}
        
    \end{tabular}
    \caption{Preference profiles illustrating why a $q$-PSC scoring rule fails ILVB.}
    \label{scoring}
\end{table}

Consider the left preference profile in Table \ref{scoring}. Note the profile contains 999 voters and thus $q=334$. Let $k=2$. Then $q$-PSC requires only that $C$ or $D$ earn a seat in the winning committee, since the solid coalition supporting $C$ and $D$ exceeds 1/3 of the voters but is less than 2/3. Thus, we calculate the aggregate scores of any committee of size two containing $C$ or $D$. If $s_2>s_1/333$, then the winning committee under a $q$-PSC scoring rule is $\{C,D\}$. However, if we remove one of the bullet votes for $B$, the quota changes to $q=333$ and $A$ is required to earn a seat by $q$-PSC since the only $q$-PSC compatible committees in this case are $\{A,C\}$ and $\{A,D\}$. Of these two committees, the $q$-PSC scoring rule chooses $\{A,C\}$, and we see a violation of ILVB.

If $s_2\le s_1/333$, then consider the right preference profile of Table \ref{scoring}. As in the previous election, $q$-PSC requires only that the winning committee contain $C$ or $D$. Because $s_2$ is so small relative to $s_1$, the winning committee is $\{B,C\}$. However, after we remove the one bullet vote for $A$ the quota decreases to 333, and $q$-PSC requires that the winner set is $\{C,D\}$.\end{proof}

Thus, even though positional scoring rules satisfy our criteria, if we attempt to combine these rules with a PSC-type constraint then the resulting methods fail, despite the fact that the $q$-PSC scoring rules do not occur sequentially like STV. We conjecture that there is no voting rule which satisfies $q$-PSC and also satisfies ILVB or IWVB but leave the proof for further research.

\section{Conclusion}\label{conclusion}

The ILVB and IWVB criteria represent an attempt to articulate reasonable fairness criteria regarding partial ballots in the real-world multiwinner setting where such ballots are very common. We argue that if we remove a set of partial ballots, only the candidates who are ranked on those ballots should be negatively affected. Our analysis shows that the voting rule of Chamberlin-Courant performs well with respect to these criteria while the expanding approvals rule performs relatively poorly, with STV-type rules falling in between.

Future research could investigate how well other voting methods perform with respect to these criteria. In particular, it would be interesting to investigate the compatibility of ILVB and IWVB with proportionality axioms other than PSC. Other work could provide better algorithms to search the real-world dataset for violations, perhaps increasing the amount of violations found. Finally, the introduction of new criteria invites a discussion about how important such criteria are, especially as compared to previously studied axioms. For example, are ILVB and IWVB as normatively desirable as candidate monotonicity or consistency \cite{EFSS}? We argue the answer is Yes, but others may feel that our criteria are less important.


\begin{thebibliography}{99}
\bibitem{A51} K. Arrow. \emph{Social Choice and Individual Values}. John
Wiley \& Sons, Inc., New York, 1st edition, 1951.


\bibitem{AL20} H. Aziz \& B. Lee. (2020). The expanding approvals rule: improving proportional representation and monotonicity. \emph{Social Choice and Welfare} \textbf{54}: 1-45.

\bibitem{AL22} H. Aziz \& B. Lee. (2022). TA characterization of proportionally representative committees. \emph{Games and Economic Behavior} \textbf{133}: 248-255.

\bibitem{BFLR} D. Baumeister, P. Faliszewski, J. Lang, \& J. Rothe. (2012). Campaigns for Lazy Voters: Truncated Ballots. In Proceedings of the International Conference on Autonomous Agents and Multiagent Systems (AAMAS’12), 57--584.


\bibitem{Bo}  B{\"o}rgers, C. (2010). {\it Mathematics of Social Choice}. Philadelphia: SIAM.


\bibitem{B82} S. Brams. (1982). The AMS nominating system is vulnerable to truncation of preferences. \emph{Notices of the American Mathematical Society} \textbf{29}: 136–138.

\bibitem{BP23} M. Brill \& J. Peters. (2023). Robust and Verifiable Proportionality Axioms for Multiwinner Voting, \url{arXiv:2302.01989}.

\bibitem{Cham-Cour} J. Chamberlin and P. Courant. (1983). Representative Deliberations and Representative Decisions: Proportional Representation and the Borda Rule, \emph{The American Political Science Review}, \textbf{77} (3): 718-733.

\bibitem{D84} M. Dummett. (1984). \emph{Voting Procedures}. Oxford University Press.

\bibitem{EFSS} E. Elkind, P. Faliszewski, P. Skowron, \& A. Slinko. (2017). Properties of multiwinner voting rules. \emph{Social Choice and Welfare} 28: 599-632. \url{https://doi.org/10.1007/s00355-017-1026-z}


\bibitem{F19} D.S. Felsenthal. (2019). On Paradoxes Afflicting Voting Procedures: Needed Knowledge Regarding Necessary and/or Sufficient Condition(s) for Their Occurrence. In: JF Laslier, H. Moulin, M. Sanver, \& W. Zwicker (eds) \emph{The Future of Economic Design}. Studies in Economic Design. Springer, Cham. \url{https://doi.org/10.1007/978-3-030-18050-8_13}.


\bibitem{FM92} D.S. Felsenthal \& Z. Moaz. (1992). Normative properties of four single-stage multi-winner electoral procedures. \emph{Behavioral Science}, \textbf{37}: 109-127.

\bibitem{FB83} P.C. Fishburn \& S. Brams. (1983). MParadoxes of Preferential Voting. \emph{Mathematics Magazine}, \textbf{56} (4): 207–214.

\bibitem{FB84} P.C. Fishburn \& S. Brams. (1984). Manipulability of voting by sincere truncation of preferences. \emph{Public Choice}, \textbf{44}: 397–410.

\bibitem{GS24} A. Graham-Squire. (2024). Conditions for Fairness Anomalies in Instant-Runoff Voting. In M. Jones, D. McCune, J. Wilson (Eds), \emph{Mathematical Analyses of Decisions, Voting and Games}. Contemporary Mathematics, Volume 795. American Mathematical Society.

\bibitem{HC15}
Heckelman, J.C. \& Chen, F. H. (2013). ``Strategy proof scoring rule lotteries for multiple winners.'' {\it Journal of Public Economic Theory 15(1)}, 103–123.

\bibitem{HWW87} I.D. Hill, B.A. Wichmann, \& D. Woodall. (1987). Algorithm 123: Single Transferable Vote by Meek's Method. \emph{The Computer Journal} \textbf{30}: 277-281. 

\bibitem{MW} D. McCune \& J. Wilson. The spoiler effect in multiwinner ranked-choice elections. \url{arXiv:2403.03228}.

\bibitem{MGS24} D. McCune \& A. Graham-Squire. (2024). Monotonicity violations in Scottish local government elections. \emph{Social Choice and Welfare}, \url{https://doi.org/10.1007/s00355-024-01522-5}.

\bibitem{M94} B.L. Meek. (1994). \emph{Voting Matters} \textbf{1}: 1-11.

\bibitem{Mi17} N. Miller. (2017). Closeness matters: monotonicity failure in IRV elections with three candidates. \emph{Public Choice}, \textbf{173}: 91-108.  \url{DOI 10.1007/s11127-017-0465-5.}

\bibitem{M19}
Miller, NR. (2019). ``Reflections on Arrow’s theorem and voting rules.''  {\it Public Choice, 179}, 113-124.


\bibitem{M95} B. Monroe. (1995). Fully Proportional Representation. \emph{The American Political Science Review}, \textbf{89} (4): 925-940.

\bibitem{N99} H. Nurmi. (1999). \emph{Voting Paradoxes and How to Deal with Them}. Springer Berlin, Heidelberg.

\bibitem{SF17} L. S\'{a}chez-Fern\'{a}ndez, E. Elikind, M. Lackner, N. Fern\'{a}ndez, J. Fisteus, P. Basanta Val, \& P. Skowron. (2017). Proportional Justified Representation. \emph{Proceedings of the AAAI Conference on Artificial Intelligence}, \textbf{31} (1).

\bibitem{T95} N. Tideman. (1995). The single transferable vote. \emph{Journal of Economic Perspectives} \textbf{9} (1): 27–38.

\bibitem{W94} D. Woodall. (1994). Properties of Preferential Election Rules. \emph{Voting Matters}, \textbf{3}: 8-15.

\bibitem{F12} D. Felsenthal. (2012). Review of Paradoxes Afflicting Procedures for Electing a Single Candidate. \emph{Studies in Choice and Welfare}. Springer Berlin, Heidelberg.

\end{thebibliography}
\end{document}